\documentclass[preprint,12pt,numbers,sort&compress]{elsarticle}
\usepackage{}
\usepackage{mathrsfs}
\usepackage{amssymb}
\usepackage{amsthm}
\usepackage{mathrsfs}
\usepackage[centertags]{amsmath}
\usepackage{amsfonts}
\usepackage{BOONDOX-cal}
\usepackage{color}
\usepackage{graphicx}
\usepackage{epsfig}
\usepackage[all]{xy}
\usepackage{ctex}
\usepackage{CJK}
\input amssym.def
\input amssym.tex

\newtheorem{theorem}{Theorem}[section]

\newtheorem{remark}{Remark}[section]
\newtheorem{assumption}{Assumption}[section]
\newtheorem{example}{Example}[section]
\numberwithin{equation}{section}

\newcommand{\mr}{\mathbb{R}}

\newcommand{\Om}{\Omega}

\newcommand{\pa}{\partial}
\allowdisplaybreaks[4]

\begin{document}

\begin{center}
{\Large \bf  A cylindrical coordinates approach concerning internal waves for the Antarctic Circumpolar Current}\\[1ex]

{Lili Fan\footnote{Email:fanlily89@126.com}$^a$,\quad  Shuge Shen\footnote{ Email:shenge990728@163.com}$^a$,
\quad}\\[1ex]

{$^a$College of Mathematics and Information Science, Henan Normal University,}\\
{Xinxiang, Henan 453007,   China}\\[1ex]

\end{center}


\begin{abstract}
In this paper, we devise a new exact and partially explicit solution to the governing equations of geophysical fluid dynamics for an inviscid and incompressible azimuth flow with a discontinuous density distribution and subjected to forcing terms in terms of cylindrical coordinates. The obtained solution represents a steady, purely azimuthal, stratified flow with an associated free surface and an interface that is suitable for describing the Antarctic Circumpolar Current. Resorting to a functional analysis, we demonstrate that the relationship between the imposed pressure at the free surface and the resulting surface deformation is well-defined and show that the continuity of the pressure along the interface generates an equation that describes implicitly the shape of the interface. Moreover, a particular example is considered to show that the interface can be determined explicitly. Finally, we derive an infinite regularity about the interface and obtain the expected monotonicity properties between the surface pressure and its distortion.
\end{abstract}

\date{}

\maketitle




\section{Introduction}

Consideration in this paper is invisible, incompressible, stratified geophysical water flows at the 45th parallel south that exhibit vertical structure, internal waves, and a preferred propagation direction. Using cylindrical coordinates in a rotating frame, we address the derivation and analysis of a new exact and partially explicit solution to the governing equations of the geophysical fluid dynamics (GFD). The solution we derive represents a steady, stratified flow propagating purely in the azimuthal direction, meaning the velocity profile, the pressure are described below and up to the free surface as a function of depth and the angle of latitude. Our obtained solution can capture the gross dynamics of the Antarctic Circumpolar Current (ACC).

It is undeniable that ACC is one of the most important and significant currents in the Earth's Ocean and it is the only current in the world that circumnavigates the globe. Serving as a link between global ocean basins, the existence of the ACC not only plays an important role in the air-sea interaction system in the Southern Ocean, but also features in the circulation of the southern hemisphere and even the whole world \cite{FCM,IR,OBVW,RHO}.

In this paper, we are devoted to investigating the behavior of the ACC affected by a discontinuous stratification, which can generate the two main fronts of the ACC \cite{PLB}. Inspired by \cite{MQ2}, we try to derive and analysis the explicit and/or exact solutions to the fully nonlinear governing equations of GFD in terms of cylindrical coordinates. The method of delivering explicit and/or exact solutions to the fully nonlinear governing equations of GFD to make study of physically realistic flows is initiated by Constantin in \cite{Co1,Co2,Co3,Co4} and Constantin and Johnson in \cite{CoJ15,CoJ17}, which has recently been adopted in numerous mathematical literature, cf. \cite{CIY,Co,H1,HM,HM1,HM2,I,Ma}. In regard to the research on the specifically azimuthal flows, it has recently been studied by Constantin and Johnson in \cite{CoJ16} for the modelling of equatorial flows and in \cite{CoJ162} for the modelling of the ACC, where the authors made the assumption that the flow is homogeneous. Motivated by the approach presented in \cite{CoJ16,CoJ162}, there have appeared some extended studies on the construction of exact solutions for modelling of stratified flows with a given density in \cite{Ba,HM,HM2,MQ1} and with a general fluid stratification in \cite{FSC,HM1,HM3,M,MQ}. Exact solutions for flows with more generally physical background can be referred to \cite{HsM,HsM1,Q} and recent studies on the ACC can be referred to \cite{CoJ162,Ha2,HsM1,MQ,MQ1,Mar}. It is worth noting that the exact azimuthal flows accommodate a discontinuous stratification have been investigated in \cite{M,MQ2} very recently  modelling the equatorial flows and the ACC in terms of spherical coordinates. Our goal here is to further advance the mathematical analysis of ACC in cylindrical coordinates, that is, to derive the exact solution applicable to the governing equations with a discontinuous stratification. Our choice of a cylindrical coordinate offers a more apparent insight into the properties of the ACC flow compared to the spherical coordinates and it simultaneously retains a relatively considerable amount of the mathematical structure of the full problem.

Undoubtedly, the appearance of discontinuous stratification complicates the analysis of ACC compared to the existing analytical investigations on the ACC. Specifically speaking, an interface, which plays the role of an internal wave, arises from the vertical stratification, that is, flows with two layers of different, nonconstant densities, where the denser layer lies below the less dense one, achieving a stable stratification. As such, in addition to the expressions of the velocity field and the pressure functions corresponding to the two layers of the fluid, the implicit descriptions of the free surface and interface are also considered intricately. To overcome this, we employ the dynamic boundary condition at the free surface to obtain the implicit relationship between the imposed pressure and the surface deformation and we use the continuity of the pressure along the interface to define the interface function. By a functional analysis, we demonstrate that the above two relations obey the implicit function theorem, leading to the desired implicit equations of free surface and interface. Furthermore, an interesting regularity property for the interface and the monotonicity properties between the free surface and the surface pressure have also been deduced using the two implicit relationships.

The remainder of this paper is organized as follows.  In Section 2, we present the governing equations for the geophysical flows in cylindrical coordinates. In Section 3, we derive exact solutions to the governing equations with the explicit azimuthal velocity and the corresponding pressure function being obtained in Section 3.1 and implicit free surface and interface solutions being obtained in Section 3.2. In Section 4, we study the regularity property of the interface, followed by a particular explicit example, and we also investigate the relations between the monotonicity of the free surface and the surface pressure.
\section{The governing equations}

In this section, we give the governing equations for geophysical fluid dynamics (GFD) in a cylindrical coordinate system, together with the boundary conditions for the free surface and a rigid bed.

The cylinder is generated by the "straightened" equator (the great circle of the sphere) which is parallel to the $z$ axis and the interior of the cylinder which amounts to the interior of the sphere is represented by standard polar coordinates. Hence, we will work in a system of right handed coordinates $(r,\theta,z)$, where $r$ denotes the distance to the center of the sphere, $\theta\in(-\pi/2,\pi/2)$ is the polar angle increasing from North to South, and the positive $z$-direction denotes azimuthal flow from West to East. Throughout this paper, we simply assume that the Antarctic Circumpolar Current (ACC) locates at
\begin{equation*}
\theta\in I_\theta:=\left[\frac \pi 4-\frac \pi {18},\frac \pi 4+\frac \pi {18}\right].
\end{equation*}

Guided by the observations in \cite{MQ2}, we will consider in our study the incompressible and inviscid fluid system which is divided into two regions
\begin{align*}
&D_1:=\{(r,\theta,z):R_2+h(\theta,z)<r<R_1+k(\theta,z)\}, \\
\text{and} \qquad &D_2:=\{(r,\theta,z):d(\theta,z)<r<R_2+h(\theta,z)\},
\end{align*}
where $R_j:=R+r_j, j=1,2$ for $0<r_2<r_1\ll R$ and $R\approx6378$ km being the radius of the earth. The subscript notations $1$, $2$ refer to values in $D_1$ and $D_2$ respectively and $j = \{1; 2\}$ for values in both domains. The functions $h(\theta,z)$ and $k(\theta,z)$ describe the unknown deviations of the interface and the free surface from their unperturbed locations at $R_2$ and $R_1$ respectively, and the function $d(\theta,z)\approx R$ denotes an impermeable solid bottom of the ocean.
To make the physical problem more convenient, we make further the following assumptions.
\begin{assumption}
(i) For $(\theta,z)\in I_\theta\times\mr$, $h\in(h_-,h_+), k\in (k_-,k_+), d\in (d_-,d_+)$.

(ii) $d<R_2+h_-$ and $R_2+h_+<R_1+k_-$.

(iii) The density $\rho$ is a depth depending discontinuous function with a jump of height $\rho_2-\rho_1$ at the interface $R_2+h$, and is defined as
\begin{equation}\label{*}
\rho(r,\theta,z)=\rho_j(r)=\rho_j+\varepsilon_j(r)=\left\{\begin{array}{l}
\rho_1+\varepsilon_1(r)\quad \text { for } r\in[R_2+h_+,R_1+k_+)\\
\rho_1+0 \quad \text { for } r\in(R_2+h,R_2+h_+) \\
\rho_2+0 \quad \text { for } r\in(R_2+h_-,R_2+h) \\
\rho_2+\varepsilon_2(r)\quad \text { for } r\in(d_-,R_2+h_-]\tag{*}
\end{array}\right.
\end{equation}
in $D_j$ for positive constants $\rho_1<\rho_2$, and $\varepsilon_j:(d_-,R_1+k_+)\rightarrow \mr$ is a smooth depth depending function satisfying $|\varepsilon_j(r)|\ll \rho_2-\rho_1$.
\end{assumption}

The governing equations in the $(r,\theta,z)$ coordinate system can then be given as the Euler's equations \cite{CoJ16,CoJ162}
\begin{equation}\label{2.1a}
\begin{cases}
u_{j,t}+u_ju_{j,r}+\frac {v_j}{r} u_{j,\theta}+w_ju_{j,z}-\frac {{v_j}^2} r-2w_j\Om \cos\theta-r\Om^2\cos^2\theta\nonumber\\
=-\frac 1 {\rho_j} p_{j,r}-g,\nonumber\\
v_{j,t}+u_jv_{j,r}+\frac {v_j} rv_{j,\theta}+w_jv_{j,z}+\frac {u_jv_j} r+2w_j\Om \sin\theta+r\Om^2\sin\theta\cos\theta\nonumber\\
=-\frac 1 {\rho_j} \frac 1 r p_{j,\theta}+G(r,\theta),\nonumber\\
w_{j,t}+u_jw_{j,r}+\frac {v_j} rw_{j,\theta}+w_jw_{j,z}+2\Om(u_j\cos\theta-v_j\sin\theta)\nonumber\\
=-\frac 1 {\rho_j} p_{j,z},\nonumber\tag{2.1a}
\end{cases}
\end{equation}
together with the equation of mass conservation
\begin{equation}\label{2.1b}
\frac 1 r \frac \pa {\pa r}(r\rho_j u_j)+\frac 1 r \frac \pa {\pa\theta}(\rho_j v_j)
+\frac \pa {\pa z}(\rho_j w_j)=0,\tag{2.1b}
\end{equation}
and the boundary conditions at the surface
\begin{align}
u_1&=w_1k_z+\frac 1 r v_1 k_\theta \qquad \text{on}\; r=R_1+k(\theta,z),\tag{2.1c}\label{2.1c}\\
p_1&=P_1(\theta,z) \qquad \qquad\quad \!\!\text{on}\; r=R_1+k(\theta,z),\tag{2.1d}\label{2.1d}
\end{align}
with boundary conditions at the interface
\begin{align}
u_j&=w_jh_z+\frac 1 r v_j h_\theta \qquad \text{on}\; r=R_2+h(\theta,z),\tag{2.1e}\label{2.1e}\\
p_1&=p_2\qquad\qquad\qquad\quad \text{on}\; r=R_2+h(\theta,z)\tag{2.1f}\label{2.1f}
\end{align}
and the boundary condition at the bed
\begin{equation}\label{2.1g}
u_2=w_2d_z+\frac 1 r v_2 d_\theta \qquad \text{on}\; r=d(\theta,z).\tag{2.1g}
\end{equation}
Here $t$ is the time, $(u_j,v_j,w_j)$ represents the velocity components, $\Omega\approx7.29\times10^{-5}$ rad/s is the rotational speed of the Earth, $p_j(r,\theta,z)$ denotes the pressure, $g\approx9.81$ m/s is the gravitational acceleration at the Earth's surface, $G(r,\theta)$ is a general body force acts in the $\theta$-direction.

\section{Exact explicit and implicit solutions}
Our aim is to seek purely azimuthal flow solutions in the region of the ACC, which means a steady flow propagating in the azimuthal direction and being independent of $z$. Thus, the velocity field is characterised by $u_j=v_j=0$ and $w_j=w_j(r,\theta)$ and moreover $p_j=p_j(r,\theta)$, $h=h(\theta), k=k(\theta), d=d(\theta)$. With this form of velocity field, the kinematic boundary conditions \eqref{2.1c}, \eqref{2.1e}, \eqref{2.1g} and the equation of mass conservation \eqref{2.1b} are automatically satisfied, whereas the Euler equations \eqref{2.1a} read
\begin{equation}\label{3.1}
\begin{cases}
-2w_j\Om \cos\theta-r\Om^2\cos^2\theta&=-\frac 1 {\rho_j} p_{j,r}-g,\\
2w_j\Om \sin\theta+r\Om^2\sin\theta\cos\theta&=-\frac 1 {\rho_j} \frac 1 r p_{j,\theta}+G(r,\theta),\\
0&=-\frac 1 {\rho_j} p_{j,z}.
\end{cases}
\end{equation}
\subsection{Explicit solutions for azimuthal velocity and pressure}
The first equation and the second equation in \eqref{3.1} can be reformulated as
\begin{align}\label{3.2}
\begin{cases}
\rho_j\Om\cos\theta(2w_j+\Om r\cos\theta)&=p_{j,r}+\rho_j g,\\
\rho_j r \Om\sin\theta(2w_j+\Om r\cos\theta)&=-p_{j,\theta}+\rho_j r G(r,\theta).
\end{cases}
\end{align}
The elimination of the pressure $p_j$ from the system \eqref{3.2} leads to
\begin{align}\label{3.3}
&r\sin\theta\left(\rho_j r \Om \cos\theta(2w_j+\Om r\cos\theta)\right)_r+\cos\theta\left(\rho_j r \Om\cos\theta(2w_j+\Om r\cos\theta)\right)_\theta\nonumber\\
&=r\cos\theta\frac \pa {\pa r}\left(\rho_j r G(r,\theta)\right)
\end{align}
in $D_j$. Employing the method of characteristics \cite{FSC,HM3}, this equation can be explicitly solved as
\begin{equation}\label{3.4}
w_j(r,\theta)=-\frac {\Om r \cos\theta} 2 +\frac 1 {2\rho_j\Om}\left[\frac {F_j(r\cos\theta)}  {r\cos\theta}+\int^{f(\theta)} _0{H_{j,r}\left(\bar{r}(s),\bar{\theta}(s)\right)}\right],
\end{equation}
where $x\rightarrow F_j(x)$ denotes some arbitrary continuously differentiable functions and
\begin{align}\label{3.5}
&f(\theta)£º=\frac 1 2 \ln \frac {1+\sin\theta}{1-\sin\theta} \qquad
H_j(r,\theta):=r \rho_j(r) G(r,\theta)\Big|_{r\in D_j}\nonumber\\
&\bar{r}(s)=r\cos\theta \cosh(s)\qquad \bar{\theta}(s)=\arcsin\left(\tanh (s)\right).
\end{align}
To determine the pressure, we infer from \eqref{3.2} and \eqref{3.4} that
\begin{align}\label{3.6}
p_{j,r}(r,\theta)&=-g\rho_j+\frac{F_j(r\cos\theta)} {r}+\cos\theta\int^{f(\theta)}_0
H_{j,r}\left(\bar{r},\bar{\theta}\right)ds
\end{align}
and
\begin{align}\label{3.7}
p_{j,\theta}(r,\theta)&=H_j(r,\theta)-\tan\theta F_j(r\cos\theta)-r\sin\theta\int^{f(\theta)}_0
H_{j,r}\left(\bar{r},\bar{\theta}\right)ds
\end{align}
Integrating with respect to $r$ in \eqref{3.6} for low layer $(j=2)$ leads to
\begin{equation}\label{3.8}
p_2(r,\theta)=\int^{r\cos\theta}_{d(\theta)\cos\theta}
\left[\frac {F_2(y)} {y}+L_2(y,\theta)\right]dy-g\int^r_{d(\theta)}\rho_2(\tilde{r},\theta)d\tilde{r}
+f_2(\theta)
\end{equation}
where $\theta\rightarrow f_2(\theta)$ is a function (to be determined) and
\begin{equation}\label{3.9}
L_j(y,\theta):=\int^{f(\theta)}_0
H_{j,r}\left(y\cosh(s),\bar{\theta}(s)\right )ds
\end{equation}
satisfying
\begin{equation}\label{3.10}
L_{j,\theta}(y,\theta)=\sec\theta H_{j,r}\left(y\sec\theta,\theta\right ).
\end{equation}
Hence
\begin{equation}\label{3.11}
\int ^{r\cos\theta}_{d(\theta)\cos\theta}L_{j,\theta}(y,\theta)dy=\int ^r_{d(\theta)}H_{j,r}(\tilde{r},\theta)d\tilde{r}=H_j(r,\theta)-H_j(d(\theta),\theta),
\end{equation}
and $f_2(\theta)$ is then given as
\begin{align}\label{3.12}
f_2(\theta)&=\int^{d(\theta)\cos\theta}_{\frac {d(\frac \pi 4)} {\sqrt2}}\left[\frac {F_2(y)} {y}+L_2(y,\tilde{\theta})\right]dy\nonumber\\
&+\int^{\theta}_{\frac \pi 4}H_2(d(\tilde{\theta}),\tilde{\theta})d{\tilde{\theta}}-g\int^{d(\theta)}_{d(\frac \pi 4)}\rho_2\left(\tilde{r}\right)d\tilde{r}.
\end{align}
For the upper layer, we have
\begin{equation}\label{3.13}
p_1(r,\theta)=\int^{r\cos\theta}_{(R_2+h(\theta))\cos\theta}
\left[\frac {F_1(y)} {y}+L_1(y,\theta)\right]dy-g\int^r_{R_2+h(\theta)}\rho_2(\tilde{r})d\tilde{r},
+f_1(h,\theta)
\end{equation}
for $r\in [R_2+h(\theta),R_1+k(\theta)] $, where
\begin{align}\label{3.14}
f_1(h,\theta)&=\int^{\theta}_{\frac \pi 4}L_1((R_2+h(\tilde{\theta}))\cos\tilde{\theta},\tilde{\theta})
\left(h^\prime(\tilde{\theta})\cos\tilde{\theta}-(R_2+h(\tilde{\theta}))\sin\tilde{\theta}\right)d\tilde{\theta}\nonumber\\
&+\int^{\theta}_{\frac \pi 4}H_1(R_2+h(\tilde{\theta}),\tilde{\theta})d\tilde{\theta}
-g\int^{\theta}_{\frac \pi 4}\rho_1(R_2+h(\tilde{\theta}))h^\prime(\tilde{\theta})d\tilde{\theta}\nonumber\\
&+\int^{\theta}_{\frac \pi 4}F_1((R_2+h(\tilde{\theta}))\cos\tilde{\theta})
\left[-\tan\tilde{\theta}+\frac {h^\prime(\tilde{\theta})} {R_2+h(\tilde{\theta})}\right]d\tilde{\theta}.
\end{align}
Without loss of generality, we have assumed that
\begin{equation}\label{3.15}
h({\frac \pi 4})=0.
\end{equation}
From the dynamic condition on the surface \eqref{2.1d}, we reach that
\begin{align}\label{3.16}
P_1(\theta)=&\int^{(R_1+k(\theta))\cos\theta}_{(R_2+h(\theta))\cos\theta}
\left[\frac {F_1(y)} {y}+L_1(y,\theta)\right]dy\nonumber\\
&-g\int^{R_1+k(\theta)}_{R_2+h(\theta)}\rho_2(\tilde{r})d\tilde{r}
+f_1(h,\theta).
\end{align}
Let $h=k\equiv0$ in \eqref{3.16}, we get that
\begin{align}\label{3.17}
P^0_1(\theta)&=\int^{\theta}_{\frac \pi 4}\left[F_1(R_2\cos\tilde{\theta})(-\tan\tilde{\theta})+H_1(R_2,\tilde{\theta})
+L_1(R_2\cos\tilde{\theta},\tilde{\theta})(-R_2\sin\tilde{\theta})\right]d\tilde{\theta}
\nonumber\\
&+\int^{R_1\cos\theta}_{R_2\cos\theta}
\left[\frac {F_1(y)} {y}+L_1(y,\theta)\right]dy-g\int^{R_1}_{R_2}\rho_2(\tilde{r})d\tilde{r},
\end{align}
which denotes the surface pressure maintaining an undisturbed free surface and interface.
\begin{assumption}\label{asm3.1}
From the argument in \cite{WN}, we make the following assumptions:

(i) The azimuthal velocity satisfies
\begin{equation}\label{3.18}
0\leq w \leq 1 ms^{-1}\quad in\quad\bar {D}_1\cup\bar{D}_2.
\end{equation}

(ii) In the region of the interface at $1000 m\leq R_1-R_2\leq 1500 m$ depth, the prescribed background density satisfies
\begin{equation}\label{3.19}
\rho_1\approx1026\, kg\, m^{-3},\quad \rho_2=\rho_1(1+\sigma),\quad 0.2\,kg\, m^{-3}\leq\sigma\rho_1\leq 0.5\, kg\, m^{-3}.
\end{equation}
\end{assumption}
\begin{remark}\label{rem3.1}
Absorbing the the centripetal and gravitational acceleration into a modified pressure $\tilde p=p+\rho g r-\frac 1 2 \rho r^2\Om^2{\cos^2\theta}$, we can deduce from \eqref{3.1} that the linear flow $w_0$ is governed by
\begin{align}\label{3.20}
\begin{cases}
-2w_0\Om\cos\theta&=-\frac 1 \rho\tilde{p_r},\\
2w_0\Om\sin\theta&=-\frac 1 {r\rho}\tilde{p_\theta}+G(r,\theta).
\end{cases}
\end{align}
As $\frac 1 {r\rho}\tilde{p_\theta}$ is small enough, we can derive from \eqref{3.20} in combination with Assumption \ref{asm3.1} that
\begin{equation}\label{3.21}
G(r,\theta)\approx 2w_0\Om\sin\theta\geq0
\end{equation}
\end{remark}
\subsection{Implicit exact solution for the interface and surface}
In this subsection, we appeal to the implicit function theorem to build the existence of unique functions $h=h(\theta)$ and $k=k(\theta)$ representing disturbances of the flat interface and surface respectively. To ensure comparability of the involved variables, we first do dimensionless processing by setting
\begin{equation}\label{3.22}
\mathcal{h}(\theta)=\frac {h(\theta)} {R_2}, \quad \mathcal{k}(\theta)=\frac {k(\theta)} {R_1},\quad\mathcal{P}_1(\theta)=\frac {P_1(\theta)} {P_{atm}}.
\end{equation}
The calculation of the implicit formula for the interface rely on the dynamic boundary condition at the interface. In fact,
\begin{equation}\label{3.23}
p_1(R_2+h(\theta),\theta)=p_2(R_2+h(\theta),\theta)
\Longleftrightarrow
\mathcal{G}(\mathcal{h})(\theta)=0,
\end{equation}
where
\begin{equation}\label{3.24}
\mathcal{G}(\mathcal{h})(\theta):=\frac {p_2((1+\mathcal{h}(\theta))R_2,\theta)-p_1((1+\mathcal{h}(\theta))R_2,\theta)} {P_{atm}}
\end{equation}
and $p_j$ is defined by \eqref{3.8} and \eqref{3.13}. Similarly, the implicit determination of the free surface depends on the dynamic condition \eqref{2.1d} and we have
\begin{equation}\label{3.25}
p_1(h(\theta),R_1+k(\theta),\theta)=P_1(\theta)
\Longleftrightarrow \mathcal{F}(\mathcal{h},\mathcal{k},\mathcal{P}_1)(\theta)=0
\end{equation}
with
\begin{equation}\label{3.26}
\mathcal{F}(\mathcal{h},\mathcal{k},\mathcal{P}_1)(\theta):=\frac {p_1((1+\mathcal{h}(\theta))R_2,(1+\mathcal{k}(\theta))R_1,\theta)} {P_{atm}}-\mathcal{P}_1(\theta)
\end{equation}
By setting $\mathcal{P}^0_1=\frac {P^0_1} {P_{atm}}$, we can combine \eqref{3.17} and \eqref{3.26} to launch that
\begin{equation}\label{3.27}
\mathcal{F}(0,0,\mathcal{P}^0_1)=0.
\end{equation}
Now, we rephrase the problem of finding $h$ and $k$ as the problem of finding solutions to the equation
\begin{equation}\label{3.28}
(\mathcal{G}(\mathcal{h}),\mathcal{F}(\mathcal{h},\mathcal{k},\mathcal{P}_1))=0.
\end{equation}
By virtue of the implicit function theorem, we obtain the following theorem.
\begin{theorem}\label{the3.1}
If Assumption \ref{asm3.1} is satisfied and $F$, $G$ are continuously differentiable in each fluid layer, then for any sufficiently small perturbation $\mathcal{P}_1$ of $\mathcal{P}^0_1$, there exists a unique tuple $(\mathcal{h},\mathcal{k})\in C^1(I_\theta)\times C(I_\theta)$ such that \eqref{3.28} holds true. Furthermore, there exists a unique continuously differentiable implicit map $\mathfrak{F}: \mathcal{P}_1\rightarrow (\mathcal{h},\mathcal{k})$ defined on
a local neighborhood of $\mathcal{P}^0_1$ such that
\[
(\mathcal{h},\mathcal{k})=\mathfrak{F}(\mathcal{P}_1) \Longleftrightarrow \eqref{3.28}\; holds.
\]
If $F$, $G$ are $n\geq2$ times continuously differentiable or infinitely differentiable, then the local map
$\mathfrak{F}$ is $c^{n}$ or $c^{\infty}$, respectively.
\end{theorem}
\begin{proof} The proof of Theorem \ref{the3.1} will be completed by the following to steps.

\vspace{0.3cm}

\textit{Step 1. The computation of $(\mathcal{G}_{\mathcal{h}}(0)\mathcal{h})(\theta), (\mathcal{F}_{\mathcal{h}}(0,0,\mathcal{P}^0_1)\mathcal{h})(\theta)$ and $(\mathcal{F}_{\mathcal{k}}(0,0,\mathcal{P}^0_1)\mathcal{k})(\theta)$.}

\vspace{0.3cm}

First, we calculate
\begin{align}\label{3.29}
&P_{atm}\cdot \lim_{s\rightarrow 0}\frac{(\mathcal{G}(s\mathcal{h})-\mathcal{G}(0))(\theta)} {s}\nonumber\\
&=\lim_{s\rightarrow 0} \frac{1}{s} \left(\int^{(1+s\mathcal{h})R_2\cos\theta}_{R_2\cos\theta}\left[\frac {F_2(y)}{y}+L_2(y,\theta)\right]dy-g\int^{(1+s\mathcal{h})R_2}_{R_2}\rho_2(\tilde{r})d\tilde{r}\right)\nonumber\\
&-\Big(\int^{\theta}_{\frac{\pi}{4}}L_1((1+s\mathcal{h})R_2\cos\tilde{\theta},\tilde{\theta})
(R_2s\mathcal{h}^\prime\cos\tilde{\theta}
-(1+s\mathcal{h})R_2\sin\tilde{\theta})d\tilde{\theta}\nonumber\\
&-\int^{\theta}_{\frac{\pi}{4}}L_1(R_2\cos\tilde{\theta},\tilde{\theta})(-R_2\sin\tilde{\theta})d\tilde{\theta}
+\int^{\theta}_{\frac{\pi}{4}}H_1((1+s\mathcal{h})R_2,\tilde{\theta})d\tilde{\theta}\nonumber\\
&-\int^{\theta}_{\frac{\pi}{4}}H_1(R_2,\tilde{\theta})d\tilde{\theta}\Big)
-\Big(\int^{\theta}_{\frac{\pi}{4}}F_1((1+s\mathcal{h})R_2\cos\tilde{\theta})\left[-\tan\tilde{\theta}
+\frac{s\mathcal{h}^\prime}{1+s\mathcal{h}^\prime}\right]d\tilde\theta\nonumber\\
&-\int^{\theta}_{\frac{\pi}{4}}F_1(R_2\cos\tilde{\theta})(-\tan\tilde{\theta})d\tilde\theta\Big)
+gR_2\int^{\theta}_{\frac{\pi}{4}}\rho_1((1+s\mathcal{h})R_2)s\mathcal{h}^\prime d{\tilde{\theta}}.
\end{align}
As
\begin{align}\label{3.30}
&\lim_{s\rightarrow 0}\frac{1}{s}\Big( \int^{\theta}_{\frac{\pi}{4}}L_1((1+s\mathcal{h})R_2\cos\tilde{\theta},\tilde{\theta})
(R_2s\mathcal{h}^\prime\cos\tilde{\theta}
-(1+s\mathcal{h})R_2\sin\tilde{\theta})d\tilde{\theta}\nonumber\\
&-\int^{\theta}_{\frac{\pi}{4}}L_1(R_2\cos\tilde{\theta},\tilde{\theta})(-R_2\sin\tilde{\theta})d\tilde{\theta}\Big)\nonumber\\
&=\lim_{s\rightarrow 0}\frac{1}{s}\Big(\int^{\theta}_{\frac{\pi}{4}} L_1((1+s\mathcal{h})R_2\cos\tilde{\theta},\tilde{\theta}) s R_{2}\left(\mathcal{h}^\prime\cos\tilde{\theta}-\mathcal{h}\sin\tilde{\theta}\right))d\tilde{\theta}\nonumber\\
&-\int^{\theta}_{\frac{\pi}{4}}R_2\sin\tilde{\theta}\left(L_1((1+s\mathcal{h})R_2\cos\tilde{\theta},\tilde{\theta})
-L_1(R_2\cos\tilde{\theta},\tilde{\theta})\right)\Big)d\tilde{\theta}\nonumber\\
&=\int^{\theta}_{\frac{\pi}{4}} L_1(R_2\cos\tilde{\theta},\tilde{\theta})d(R_2\mathcal{h}\cos\tilde{\theta})-\lim_{s\rightarrow 0}\frac{1}{s}\int^{\theta}_{\frac{\pi}{4}}R_2\sin\tilde{\theta}
\left(L_{1,y}(\xi_1,\tilde{\theta})s\mathcal{h}R_2\cos\tilde{\theta}\right)d\tilde{\theta}\nonumber\\
&:=A-R_2^2\int^{\theta}_{\frac{\pi}{4}}L_{1,y}(R_2\cos\tilde{\theta},\tilde{\theta})\mathcal{h}
\sin\tilde{\theta}\cos\tilde{\theta}d\tilde{\theta},
\end{align}
where $R_2\cos\tilde{\theta}\leq\xi_1\leq(1+s\mathcal{h})R_2\cos\tilde{\theta}$ and with
\begin{align}\label{3.31}
A&=L_1(R_2\cos\theta,\theta)R_2\mathcal{h}\cos\theta
+R_2^2\int^{\theta}_{\frac{\pi}{4}}L_{1,y}(R_2\cos\tilde{\theta},\tilde{\theta})\mathcal{h}
\sin\tilde{\theta}\cos\tilde{\theta}d\tilde{\theta}\nonumber\\
&-R_2\int^{\theta}_{\frac{\pi}{4}}L_{1,\theta}(R_2\cos\tilde{\theta},\tilde{\theta})
\mathcal{h}\cos\tilde{\theta}d\tilde{\theta},
\end{align}
and
\begin{align}\label{3.32}
&\lim_{s\rightarrow 0} \frac{1}{s}\int^{\theta}_{\frac{\pi}{4}}\left(H_1((1+s\mathcal{h})R_2,\tilde{\theta})
-H_1(R_2,\tilde{\theta})\right)d\tilde{\theta}\nonumber\\
&=\lim_{s\rightarrow 0}\frac{1}{s} \int^{\theta}_{\frac{\pi}{4}}H_{1,r}(\xi_2,\tilde{\theta})s\mathcal{h}R_2 d\tilde{\theta}
\nonumber\\
&=\int^{\theta}_{\frac{\pi}{4}}H_{1,r}(R_2,\tilde{\theta})\mathcal{h}R_2 d\tilde{\theta}=R_2\int^{\theta}_{\frac{\pi}{4}}L_{1,\theta}(R_2\cos\tilde{\theta},\tilde{\theta})
\mathcal{h}\cos\tilde{\theta}d\tilde{\theta}.
\end{align}
for $R_2\leq \xi_2 \leq (1+s\mathcal{h})R_2$ and the last equality is valid by \eqref{3.10}. Furthermore,
\begin{align}\label{3.33}
&\lim_{s\rightarrow 0} \frac{1}{s} \left(\int^{(1+s\mathcal{h})R_2\cos\theta}_{R_2\cos\theta}\left[\frac {F_2(y)}{y}+L_2(y,\theta)\right]dy-g\int^{(1+s\mathcal{h})R_2}_{R_2}\rho_2(\tilde{r})d\tilde{r}\right)\nonumber\\
&=\lim_{s\rightarrow 0}\left[\frac {F_2((1+s\mathcal{h})R_2\cos\theta)}{(1+s\mathcal{h})R_2\cos\theta}
+L_2((1+s\mathcal{h})R_2\cos\theta,\theta)\right]\mathcal{h}R_2\cos\theta\nonumber\\
&-g\rho_2((1+s\mathcal{h})R_2)\mathcal{h}R_2\nonumber\\
&=\left(F_2(R_2\cos\theta)+L_2(R_2\cos\theta,\theta)R_2\cos\theta-gR_2\rho_2\right)\mathcal{h}(\theta),
\end{align}
and
\begin{align}\label{3.34}
&\lim_{s\rightarrow 0} \frac{1}{s}
\int^{\theta}_{\frac{\pi}{4}}\left(F_1((1+s\mathcal{h})R_2\cos\tilde{\theta})\left[-\tan\tilde{\theta}
+\frac{s\mathcal{h}^\prime}{1+s\mathcal{h}^\prime}\right]
-F_1(R_2\cos\tilde{\theta})(-\tan\tilde{\theta})\right)d\tilde\theta\nonumber\\
&=\lim_{s\rightarrow 0}\frac{1}{s}\int^{\theta}_{\frac{\pi}{4}}\left(-\tan\tilde{\theta}F_{1,x}(\xi_3)s\mathcal{h}R_2\cos\tilde{\theta}
+F_1((1+s\mathcal{h})R_2\cos\tilde{\theta})\frac{s\mathcal{h}^\prime}{1+s\mathcal{h}^\prime}\right)
d\tilde\theta\nonumber\\
&=\int^{\theta}_{\frac{\pi}{4}}\left(-\tan\tilde{\theta}F_{1,x}(R_2\cos\tilde{\theta})\mathcal{h}R_2\cos\tilde{\theta}
+F_1(R_2\cos\tilde{\theta})\mathcal{h}^\prime\right)d\tilde\theta\nonumber\\
&=F_1(R_2\cos\tilde{\theta})\mathcal{h}\Big|_{\frac{\pi}{4}}^{\theta}
=F_1(R_2\cos\theta)\mathcal{h},
\end{align}
with $R_2\cos\tilde{\theta}\leq \xi_3\leq (1+s\mathcal{h})R_2\cos\tilde{\theta}$, and
\begin{equation}\label{3.35}
\lim_{s\rightarrow 0} \frac{1}{s} gR_2\int^{\theta}_{\frac{\pi}{4}}\rho_1((1+s\mathcal{h})R_2)s\mathcal{h}^\prime d{\tilde{\theta}}=gR_2\rho_1\mathcal{h}.
\end{equation}
Substituting \eqref{3.35}-\eqref{3.35} into \eqref{3.29}, we obtain that
\begin{align}\label{3.36}
&P_{atm}\cdot \lim_{s\rightarrow 0}\frac{(\mathcal{G}(s\mathcal{h})-\mathcal{G}(0))(\theta)} {s}\nonumber\\
&=\Big(F_2(R_2\cos\theta)-F_1(R_2\cos\theta)-gR_2(\rho_2-\rho_1)
+\big(L_2(R_2\cos\theta,\theta)\nonumber\\
&-L_1(R_2\cos\theta,\theta)\big)R_2\cos\theta\Big)\mathcal{h}(\theta),
\end{align}
then
\begin{align}\label{3.37}
(\mathcal{G}_{\mathcal{h}}(0)\mathcal{h})(\theta)&=\frac 1
{P_{atm}}\Big(F_2(R_2\cos\theta)-F_1(R_2\cos\theta)-gR_2(\rho_2-\rho_1)\nonumber\\
&+\big(L_2(R_2\cos\theta,\theta)-L_1(R_2\cos\theta,\theta)\big)R_2\cos\theta\Big)\mathcal{h}(\theta).
\end{align}
On the other hand, we compute similarly as \eqref{3.29} to get that for all $\theta\in I_{\theta}$
\begin{align}\label{3.44}
&P_{atm}(\mathcal{F}_{\mathcal{h}}(0,0,\mathcal{P}^0_1)\mathcal{h})(\theta)\nonumber\\
&=\lim_{s\rightarrow0}\frac {p_1((1+s\mathcal{h})R_2,R_1,\theta)-p_1(R_2,R_1,\theta)} {s}\nonumber\\
&=(F_1(R_2\cos\theta)+L_1(R_2\cos\theta,\theta)R_2\cos\theta-gR_2\rho_1\nonumber\\
&-F_1(R_2\cos\theta)-L_1(R_2\cos\theta,\theta)R_2\cos\theta+gR_2\rho_1)\mathcal{h}(\theta)\nonumber\\
&=0,
\end{align}
and
\begin{align}\label{3.45}
&P_{atm}(\mathcal{F}_{\mathcal{k}}(0,0,\mathcal{P}^0_1)\mathcal{k})(\theta)\nonumber\\
&=\lim_{s\rightarrow0}\frac {p_1(R_2,(1+s\mathcal{k})R_1,\theta)-p_1(R_2,R_1,\theta)} {s}\nonumber\\
&=(F_1(R_1\cos\theta)+L_1(R_1\cos\theta,\theta)R_1\cos\theta-gR_1\rho_1(R_1))\mathcal{k}(\theta).
\end{align}

\vspace{0.3cm}

\textit{Step 2. The evaluation of $(\mathcal{G}_{\mathcal{h}}(0)\mathcal{h})(\theta), (\mathcal{F}_{\mathcal{h}}(0,0,\mathcal{P}^0_1)\mathcal{h})(\theta)$ and $(\mathcal{F}_{\mathcal{k}}(0,0,\mathcal{P}^0_1)\mathcal{k})(\theta)$.}

\vspace{0.3cm}

We set $\bar{F}_j(\theta)=F_j(R_2\cos\theta)$, $\bar{L}_j(\theta)=L_j(R_2\cos\theta,\theta)R_2\cos\theta$ and $\bar{w}_j(\theta)=w_j(R_2,\theta)$ for $j=1,2$. Due to \eqref{3.4} and Assumption \ref{asm3.1}, we reach
\begin{align}\label{3.38}
&(\bar{F}_2-\bar{F}_1+\bar{L}_2-\bar{L}_1)(\theta)\nonumber\\
&=(F_2(R_2\cos\theta)+L_2 R_2\cos\theta )-(F_1(R_2\cos\theta)+L_1 R_2\cos\theta )\nonumber\\
&=(\bar{w}_2+\frac {\Om R_2\cos\theta} {2})2\rho_2\Om R_2\cos\theta-(\bar{w}_1+\frac {\Om R_2\cos\theta} {2})2\rho_1\Om R_2\cos\theta\nonumber\\
&=(\bar{w}_2+\frac {\Om R_2\cos\theta} {2})2\rho_1\Om R_2\cos\theta+(\bar{w}_2+\frac {\Om R_2\cos\theta} {2})2\rho_1\sigma\Om R_2\cos\theta\nonumber\\
&-(\bar{w}_1+\frac {\Om R_2\cos\theta} {2})2\rho_1\Om R_2\cos\theta\nonumber\\
&=2\rho_1\Om R_2\cos\theta(\bar{w}_2-\bar{w}_1)+2\rho_1\sigma\Om R_2\cos\theta(\bar{w}_2+\frac {\Om R_2\cos\theta} {2})\nonumber\\
&\leq R_2(2\rho_1\Om\bar{w}_2+2\rho_1\sigma\Om\bar{w}_2+R_2\rho_1\sigma\Om^2),\nonumber\\
&<0.17R_2 kgm^{-1}s^{-2}
\end{align}
for all $\theta\in I_{\theta}$. Meanwhile,
\begin{equation}\label{3.39}
gR_2(\rho_2-\rho_1)=gR_2\rho_1\sigma >1.9R_2 kgm^{-1}s^{-2}
\end{equation}
From the inequalities \eqref{3.38} and \eqref{3.39}, we can find a constant $\alpha<0$ such that
\begin{equation}\label{3.40}
(\bar{F}_2-\bar{F}_1+\bar{L}_2-\bar{L}_1)(\theta)-gR_2(\rho_2-\rho_1)\leq\alpha,
\end{equation}
holds right for all $\theta\in I_{\theta}$. On the other hand, we obtain that for all $\theta\in I_{\theta}$
\begin{align}\label{3.41}
&F_1(R_1\cos\theta)+L_1(R_1\cos\theta,\theta)R_1\cos\theta-gR_1\rho_1(R_1)\nonumber\\
&=(w_1(R_1,\theta)+\frac {\Om R_1\cos\theta} {2})2\rho_1(R_1)\Om R_1\cos\theta-gR_1\rho_1(R_1)\nonumber\\
&=\rho_1(R_1)[2\Om R_1\cos\theta w_1(R_1,\theta)+(\Om R_1\cos\theta)^2-gR_1]\nonumber\\
&\leq\rho_1(R_1)[2\Om R_1w_1(R_1,\theta)+(\Om R_1)^2-gR_1]\nonumber\\
&<-\rho_1(R_1)\cdot6\cdot10^{-7} kgm^{-1}s^{-2}
\end{align}
As noted above, we can infer that there is $\beta<0$ such that
\begin{equation}\label{3.42}
F_1(R_1\cos\theta)+L_1(R_1\cos\theta,\theta)R_1\cos\theta-gR_1\rho_1(R_1)\leq\beta
\end{equation}
for all $\theta\in I_{\theta}$. The inequalities \eqref{3.40} and \eqref{3.42} guarantee that the map
$$\mathcal{h}\rightarrow\mathcal{G}_{\mathcal{h}}(0)\mathcal{h}$$
is a linear topological automorphism from $C^1(I_{\theta})$ to $C^1(I_{\theta})$, and the map
$$\mathcal{k}\rightarrow\mathcal{F}_{\mathcal{k}}(0,0,\mathcal{P}^0_1)\mathcal{k}$$
is a linear topological automorphism from $C(I_{\theta})$ to $C(I_{\theta})$. To sum up, we can deduce that
\begin{equation}\label{3.43}
\begin{aligned}
(\mathcal{G},\mathcal{F})_{(\mathcal{h},\mathcal{k})}(0,0,\mathcal{P}^0_1)
&=\begin{pmatrix}
{\mathcal{G}_{\mathcal{h}}(0)}&{\mathcal{G}_{\mathcal{k}}(0)}\\
{\mathcal{F}_{\mathcal{h}}(0,0,\mathcal{P}^0_1)} & {\mathcal{F}_{\mathcal{k}}(0,0,\mathcal{P}^0_1)}
\end{pmatrix}\\
&=\begin{pmatrix}
{\mathcal{G}_{\mathcal{h}}(0)}&0\\
0 & {\mathcal{F}_{\mathcal{k}}(0,0,\mathcal{P}^0_1)}
\end{pmatrix}
\end{aligned}
\end{equation}
constitutes a linear topological automorphism of $C^1(I_{\theta})\times C(I_{\theta})$. Therefore, combining \eqref{3.27} with the above arguments in the two steps, we get the results for the existence of free surfaces and interfaces by the implicit function theorem.
\end{proof}
\section{Qualitative Analysis Of Solutions}
The aim of this section is to give qualitative results for the interface and the free surfaces, as well as a concrete practice that explicitly describes the interface.
\begin{theorem}\label{the4.1}(Smoothness of the interface).
If Assumption \ref{asm3.1} is satisfied and $F$, $G$ are infinitely differentiable in each fluid layer, and on the basis of Theorem \ref{the3.1}, set $\mathcal{P}_1$ is a given small enough perturbation of $\mathcal{P}^0_1$ associating with $\mathcal{h}\in C^1(\theta)$, then $\mathcal{h}\in C^\infty(\theta)$; if $F$, $G$ are n times continuously differentiable for $n\geq1$, then $\mathcal{h}\in C^{n+1}(\theta)$.
\end{theorem}
\begin{proof}
From Theorem \ref{the3.1}, $(\mathcal{G}(\mathcal{h}))(\theta)=0$ holds true for all $\theta\in I_\theta$, and thus a derivation with respect to $\theta$ leads to
\begin{equation}\label{4.1}
(A_1(\theta)+A_2)\mathcal{h}^\prime(\theta)+B(\theta)=0,
\end{equation}
where
\begin{align}\label{4.2}
A_1(\theta):=\frac {\bar{F}_2(\theta)-\bar{F}_1(\theta)+\bar{L}_2(\theta)-\bar{L}_1(\theta)} {1+\mathcal{h}},
\end{align}
\begin{align}\label{4.3}
A_2:=-gR_2(\rho_2-\rho_1),
\end{align}
\begin{align}\label{4.4}
B(\theta)
=[\bar{F}_2(\theta)-\bar{F}_1(\theta)+\bar{L}_2(\theta)-\bar{L}_1(\theta)](-\tan\theta)+\bar{H}_2(\theta)-\bar{H}_1(\theta),
\end{align}
\begin{align}\label{4.5}
\bar{F}_j(\theta)=F_j((1+\mathcal{h}(\theta))R_2\cos\theta),
\end{align}
\begin{align}\label{4.6}
\bar{H}_j(\theta)=H_j((1+\mathcal{h}(\theta))R_2,\theta),
\end{align}
\begin{align}\label{4.7}
\bar{L}_j(\theta)=L_j((1+\mathcal{h}(\theta))R_2\cos\theta,\theta)(1+\mathcal{h}(\theta))R_2\cos\theta.
\end{align}
Now, we certify that $A(\theta):=A_1+A_2<0$ for all $\theta\in I_\theta$. To this end, by setting $\bar{w}_j(\theta):=w_j((1+\mathcal{h})R_2,\theta)$ and computing similarly as \eqref{3.38}, we have
\begin{align}\label{4.8}
A_1(\theta)
&=2\rho_1\Om R_2\cos\theta(\bar{w}_2-\bar{w}_1)+2\rho_1\sigma\Om R_2\cos\theta\left(\bar{w}_2+\frac{\Om R_2(1+\mathcal{h})\cos\theta}{2}\right)\nonumber\\
&\leq2\rho_1\Om R_2\bar{w}_2+2\rho_1\sigma\Om R_2\bar{w}_2+\rho_1\sigma\Om^2R_2^2(1+\mathcal{h})\nonumber\\
&=R_2\left(2\rho_1\Om\bar{w}_2+2\rho_1\sigma\Om\bar{w}_2+\rho_1\sigma\Om^2R_2(1+\mathcal{h})\right)\nonumber\\
&\leq2R_2(2\rho_1\Om\bar{w}_2+2\rho_1\sigma\Om\bar{w}_2+\rho_1\sigma\Om^2R_2)\nonumber\\
&<0.34R_2kgm^{-1}s^{-2},
\end{align}
and by \eqref{3.39}, $A_2=-gR_2\rho_1\sigma<-1.9R_2 kgm^{-1}s^{-2}$. Hence, there exists some $\gamma <0$, such that
\begin{align}\label{4.9}
A(\theta)\leq\gamma \quad\text {for all}\; \theta\in I_\theta.
\end{align}

When $F$ and $G$ belong to $C^\infty$, respectively, $C^n$, it is correct that the regularity of $\mathcal{h}$ can be transferred to $A$ and $B$, in other words, $A, B\in C^1(I_\theta)$ is satisfied. On the other hand, \eqref{4.9} implies that $\mathcal{h}'=-\frac{B}{A}$, so we get $\mathcal{h}\in C^2(I_\theta)$. By induction, it follows that $\mathcal{h}\in C^\infty(I_\theta)$, respectively, $\mathcal{h}\in C^{n+1}(I_\theta)$. The proof is completed.
\end{proof}
\begin{example}
Given a priori explicit density, we now reveal an explicit solution in terms of the velocity field, the pressure, and the free surface and the interface. To be more precisely, we consider the following scenario: the density $\rho_2(r)$ throughout the whole lower layer is a constant, and for the upper layer, we identify from \eqref{*} that
\begin{align}\label{4.10}
\rho_1(r)=\rho_1+\varepsilon_1(r)=\rho_1+(\rho_2-\rho_1-a_1r)=\rho_2-a_1r,
\end{align}
where $a_1:=\frac {\rho_2-\rho_1}{R_2+h_+}$ is a constant satisfying the density shown in Assumption \ref{asm3.1}. Meanwhile, according to the proof in Theorem \ref{the4.1}, we deduce that
\begin{align}\label{4.11}
A_2=-gR_2(\rho_2(R_2+h(\theta))-\rho_1(R_2+h(\theta)))=-ga_1R_2^2(1+\mathcal{h}(\theta)).
\end{align}
We begin by specifying the velocity field. To this end,  we particularize
the functions $F_1$ and $F_2$ appearing in formula of the azimuthal velocity $w_j$ in \eqref{3.4}. That is, we set
\begin{align}\label{4.12}
F_1(x)=\alpha_1\rho_2\Om^2x^2,\quad F_2(x)=\alpha_2\rho_2\Om^2x^2
\end{align}
where $\alpha_1,\alpha_2$ are dimensionless constants satisfying $\alpha_2>\alpha_1$. Hence
\begin{align}\label{4.13}
\bar{F}_2(\theta)-\bar{F}_1(\theta)=(\alpha_2-\alpha_1)\rho_2\Om^2R_2^2(1+\mathcal{h})^2\cos^2\theta.
\end{align}
In addition, when $G$ is given as in formula \eqref{3.21}, we have
\begin{align}\label{4.14}
\begin{cases}
H_1(r,\theta)=\rho_1rG=2\Om w_0(\rho_2r-a_1r^2)\sin\theta,\\
H_2(r,\theta)=\rho_2rG=2\Om w_0\rho_2r\sin\theta,
\end{cases}
\end{align}
therefore
\begin{align}\label{4.15}
\bar{H}_2(\theta)-\bar{H}_1(\theta)=2a_1\Om w_0R_2^2(1+\mathcal{h})^2\sin\theta.
\end{align}
Moreover, we infer that
\[
H_{2,r}(r,\theta)-H_{1,r}(r,\theta)=\frac \pa {\pa r}(2\Om w_0a_1r^2\sin\theta)=4\Om w_0a_1r\sin\theta,
\]
and hence
\begin{align}
&H_{2,r}\left(r\cos\theta\frac {e^s+e^{-s}}{2},\bar{\theta}(s)\right)-H_{1,r}\left(r\cos\theta\frac {e^s+e^{-s}}{2},\bar{\theta}(s)\right)\nonumber\\
&=4\Om w_0a_1r\cos\theta\frac {e^s+e^{-s}}{2}\sin\bar{\theta}(s)\nonumber\\
&=2\Om w_0a_1r\cos\theta(e^s+e^{-s})\frac {e^{2s}-1}{e^{2s}+1}\nonumber\\
&=2\Om w_0a_1r\cos\theta (e^s-e^{-s}).\nonumber
\end{align}
From
\[
L_j(y,\theta)=\int^{f(\theta)}_{0}H_{j,r}\left(y\frac{e^s+e^{-s}}{2},\bar{\theta}(s)\right)ds,
\]
we find
\begin{align}\label{4.16}
L_2(r\cos\theta,\theta)-L_1(r\cos\theta,\theta)=4\Om w_0a_1r(1-\cos\theta),
\end{align}
thus
\begin{align}\label{4.17}
\bar{L}_2(\theta)-\bar{L}_1(\theta)=
4\Om w_0a_1{R_2}^2{(1+\mathcal{h})}^2\cos\theta(1-\cos\theta).
\end{align}
Based on this, the equation \eqref{4.1} can be written as
\begin{align}\label{4.18}
&\left((\alpha_2-\alpha_1)\rho_2\Om^2(-\sin\theta\cos\theta)+4\Om w_0a_1(-\sin\theta(1-\cos\theta))+2a_1\Om w_0\sin\theta\right)(1+\mathcal{h}(\theta))\nonumber\\
&+\left((\alpha_2-\alpha_1)\rho_2\Om^2\cos^2\theta+4\Om w_0a_1\cos\theta(1-\cos\theta)-ga_1\right)\mathcal{h}^\prime(\theta)=0,
\end{align}
which equals to
\begin{align}\label{4.19}
\frac{\mathcal{h}^\prime(\theta)}{1+\mathcal{h}(\theta)}=\frac{\tilde{E}(\theta)}{E(\theta)},
\end{align}
where
\begin{align}\label{4.20}
&\tilde{E}(\theta)=(\alpha_2-\alpha_1)\rho_2\Om^2(\sin\theta\cos\theta)+4\Om w_0a_1(\sin\theta(1-\cos\theta))-2a_1\Om w_0\sin\theta,\nonumber\\
&E(\theta)=(\alpha_2-\alpha_1)\rho_2\Om^2\cos^2\theta+4\Om w_0a_1\cos\theta(1-\cos\theta)-ga_1.
\end{align}
Noting that $E^\prime (\theta)=-2\tilde{E}(\theta)$, we get from \eqref{4.19}
\[
\frac {1}{1+\mathcal{h}(\theta)}=C\sqrt{E(\theta)}
\]
Recalling $\mathcal{h}(\frac{\pi}{4})=0$, we obtain $C=1/\sqrt{E(\frac{\pi}{4})}$ and hence
\begin{align}\label{4.21}
1+\mathcal{h}(\theta)=\sqrt{\frac{E(\frac{\pi}{4})}{E(\theta)}}.
\end{align}
Submitting \eqref{4.12} and \eqref{4.14} into \eqref{3.4}, we can obtain the velocity field $w_j$, and submitting \eqref{4.12}, \eqref{4.14} and \eqref{4.21} into \eqref{3.8} and \eqref{3.13}, we can get the pressure $p_j$. Finally, the free surface can be solved by \eqref{3.16} and \eqref{4.21}.
\end{example}
\begin{theorem}(Monotonicity preserving properties). If Assumption \ref{asm3.1} is satisfied and $F_j, G$ are infinitely differentiable in each fluid layer, and on the basis of Theorem \ref{the3.1}, set $\mathcal{P}_1\in C^1(I_\theta)$ of $\mathcal{P}^0_1$ being a given small enough perturbation associated with $\mathcal{k}\in C(I_\theta)$, then $\mathcal{k}\in C^1(I_\theta)$ for all $\theta\in I_\theta$. Moreover, for $\theta\in I_\theta$, it holds that
\begin{align}\label{4.22}
\mathcal{P}_1'(\theta)<0\;\; \text{if}\;\; \mathcal{k}'(\theta)\geq 0,
\end{align}
and
\begin{align}\label{4.23}
\mathcal{k}'(\theta)<0\;\; \text{if}\;\; \mathcal{P}_1'(\theta)\geq0.
\end{align}
\end{theorem}
\begin{proof}
We first remark that utilising an iterative bootstrapping procedure, smoothness properties of $\mathcal{P}_1$ can be transferred to $\mathcal{k}$. Hence, we differentiate with respect to $\theta$ in \eqref{3.16} to reach that
\begin{align}\label{4.24}
&P_{atm}\mathcal{P}'_1(\theta)\nonumber\\
&=\left[\frac{F_1((R_1+k(\theta))\cos\theta)}{(R_1+k(\theta))\cos\theta}+L_1((R_1+k(\theta))\cos\theta,\theta)\right][k'(\theta)\cos\theta-(R_1+k(\theta))\sin\theta]\nonumber\\
&-g\rho_1(R_1+k(\theta))k'(\theta)+H_1(R_1+k(\theta),\theta)\nonumber\\
&=\left(\bar{w}_1+\frac{\Om (R_1+k(\theta))\cos\theta}{2}\right)2\rho_1(R_1+k(\theta))\Om[k'(\theta)\cos\theta-(R_1+k(\theta))\sin\theta]\nonumber\\
&-g\rho_1(R_1+k(\theta))k'(\theta)+H_1(R_1+k(\theta),\theta)\nonumber\\
&=\left(2\bar{w}_1\Om \cos\theta+\Om^2\cos^2\theta R_1(1+\mathcal{k}(\theta))-g\right)R_1\rho_1(R_1(1+\mathcal{k}(\theta)))\mathcal{k}'(\theta)\nonumber\\
&+p_{1,\theta}(R_1(1+\mathcal{k}(\theta),\theta)
\end{align}
for all $\theta \in I _\theta$, where $\bar{w}_1(\theta):=w_1(R_1(1+\mathcal{k}),\theta)$. And thus
\begin{align}\label{4.25}
\frac{P_{atm}}{R_1\rho_1(R_1(1+\mathcal{k}(\theta)))}\mathcal{P}'_1(\theta)
&=\left(2\bar{w}_1\Om \cos\theta+\Om^2\cos^2\theta R_1(1+\mathcal{k}(\theta))-g\right)\mathcal{k}'(\theta)\nonumber\\
&+\frac 1 {R_1\rho_1(R_1(1+\mathcal{k}(\theta)))}p_{1,\theta}(R_1(1+\mathcal{k}(\theta),\theta)
\end{align}
According to the discussion in \cite{CoJ162}, the meridional pressure gradient is typically relatively small compared to the radial gradient. Then by the first equation in \eqref{3.1}, it holds that
\begin{equation}\label{4.26}
\frac 1 {R_1}\frac {p_{1,\theta}}{\rho}\leq \frac 1 {R_1}\frac {p_{1,r}}{\rho}=\frac 1 {R_1}\left(2\bar{w}_1\Om \cos\theta+\Om^2\cos^2\theta R_1(1+\mathcal{k}(\theta))-g\right),
\end{equation}
implying
\begin{align}\label{4.27}
&\frac{P_{atm}}{R_1\rho_1(R_1(1+\mathcal{k}(\theta)))}\mathcal{P}'_1(\theta)\nonumber\\
&\leq(\mathcal{k}'(\theta)+\frac 1 {R_1})\left(2\bar{w}_1\Om \cos\theta+\Om^2\cos^2\theta R_1(1+\mathcal{k}(\theta))-g\right).
\end{align}
A similar computation as \eqref{3.41} in the proof of Theorem \ref{the3.1} shows that there exists a constant $\delta<0$ such that for all $\theta\in I_\theta$
\begin{align}\label{4.28}
2\bar{w}_1\Om \cos\theta+R_1\Om^2\cos^2\theta(1+\mathcal{k})-g\leq\delta<0.
\end{align}
Thus, assuming $\mathcal{P}'_1(\theta)\geq0$, we see from \eqref{4.27} that $\mathcal{k}'(\theta)\leq-\frac 1 { R_1}<0$, which shows \eqref{4.23}. The proof of the monotonicity properties \eqref{4.22} follows by the same arguments.
\end{proof}

\vspace{0.5cm}
\noindent {\bf Acknowledgements.}
The work of Fan is partially supported by a NSFC Grant No. 11701155 and the NSF of Henan Normal University Grant No. 2021PL04.

\end{document}